\documentclass[journal]{IEEEtran}
\usepackage{ifpdf}
 \ifpdf
 \else
 \fi
\usepackage{cite}
\ifCLASSINFOpdf
   \usepackage[pdftex]{graphicx}
   \usepackage{epstopdf}
\else
   \usepackage[dvips]{graphicx}
\fi
\usepackage[cmex10]{amsmath}
\usepackage{amssymb}
\usepackage{amsthm}
\begin{document}
\title{Unitary Query for the $M \times L \times N$ MIMO Backscatter RFID Channel}

\author{Chen~He,~
        Z.~Jane~Wang,~\IEEEmembership{Senior Member,~IEEE},~
        and Victor~C.M.~Leung, ~\IEEEmembership{Fellow,~IEEE}
                \thanks{The authors are with \IEEEauthorblockA{Department of Electrical and Computer Engineering, University of British Columbia, Vancouver, BC, V6T 1Z4 Canada. Emails: \{chenh, zjanew, vleung\}@ece.ubc.ca.} }
}% <-this % stops a space
%\thanks{M. Shell is with the Department
%of Electrical and Computer Engineering, Georgia Institute of Technology, Atlanta,
%\thanks{Manuscript received April 19, 2005; revised January 11, 2007.}}


\markboth{}%
{Shell \MakeLowercase{\textit{et al.}}: Bare Demo of IEEEtran.cls for Journals}



\maketitle

\begin{abstract}
%\boldmath
A MIMO backscatter RFID system consists of three operational ends: the query end (with $M$ reader transmitting antennas), the tag end (with $L$ tag antennas) and the receiving end (with $N$ reader receiving antennas). Such an $M \times L \times N$ setting in RFID can bring spatial diversity and has been studied for STC at the tag end.
Current understanding of the query end is that it is only an energy provider for the tag and query signal designs cannot improve the performance. However, we propose a novel \textit{unitary query} scheme, which creates time diversity \emph{within channel coherent time} and can yield \emph{significant} performance improvements. To overcome the difficulty of evaluating the performance when the unitary query is employed at the query end and STC is employed at the tag end, we derive a new measure based on the ranks of certain carefully constructed matrices. The measure implies that the unitary query has superior performance. Simulations show that the unitary query can bring $5-10$ dB gain in mid SNR regimes. In addition, the unitary query can also improve the performance of single-antenna tags significantly, allowing employing low complex and small-size single-antenna tags for high performance. This improvement is unachievable for single-antenna tags when the conventional uniform query is employed.
\end{abstract}

%Conventional measures of the physical layer performance, such as the asymptotic pairwise error probability and the diversity order, can not be obtained analytically in backscatter RFID channels with employing both space-time code and unitary query schemes. We thus propose a new performance measure to overcome such a difficulty of conventional measures, and we prove why the proposed unitary query has superior BER performances.


% \fi
\IEEEpeerreviewmaketitle

\newtheorem{Theorem}{Theorem}
\newtheorem{Lemma}{Lemma}
\newtheorem{Proposition}{Proposition}
\section{Introduction}
Radio-frequency identification (RFID) is a wireless communication technology that allows an
object to be identified automatically and does not require LOS transmission \cite{Want2006}.
It is one important infrastructure of the internet of things, and adds significant values in many applications, such as inventory systems, product tracking, access control, libraries, museums, sports and social networks.
An RFID system includes three major components: RFID readers  (also known as interrogators), RFID tags (also known as labels), and RFID software or RFID middleware \cite{Klaus2003}. An RFID tag is a small electronic device that has a unique ID. It transmits data over the air in response to interrogation by an RFID reader.
Depending on power supplying methods, the RFID tags can be categorized into passive, active, and semi-active tags.  An active tag utilizes its internal battery to continuously power its RF communication circuitry, while a passive RFID tag has no internal power supply and relies on RF energy transferred from the reader to the tag. A semi-passive tag is powered by both its internal battery and RF energy from the reader.

Most RFID tags deployed are based on  \emph{backscatter modulation}, which does not require the modulated signal to be amplified and retransmitted, and thus the RF tags can be made extraordinary small and inexpensive. By the principal of
backscatter modulation, the RF tag simply scatters a portion of the
incident continuous wave signals from the reader transmitter
back to the reader receiver using load modulation \cite{Kim2003}. Such signals sent from the reader transmitters are known as \emph{query signals}. The backscatter RFID can operate at ultra-high frequency
(UHF) at $860-960$ MHz, $2.45$ GHz and $5.8$ GHz with the
operating range of the order of $10$ meters.
%For the passive technology, tags capture the EM waves propagating from a dipole antenna attached to the reader. A smaller dipole antenna in the tag receives the energy as an alternating potential difference that will result in an %accumulation of energy to power its circuit. This approach is referred as \emph{backscatter} principle \cite{Want2006}.
%In physics, backscatter is the reflection of waves, particles, or signals back to the direction from which they came.
%A backscatter RFID signal modulation procedure is shown in Fig.\ref{Fig: Modulation_RFID}. The RF reader transmitter first broadcasts an unmodulated carrier signal, which is also known as query signal, then the RF tag conveys information (i.e. the ID of the tag) to the reader by simply reflecting a the query signal from the reader transmitter back to the reader receiver using load modulation \cite{Kim2003} \cite{Griffin2009}. The ID information of the RF tag depends on the reflection coefficient of the tag antenna load  which is changed by switching the RF tag antenna load between different states \cite{Griffin2009}.
%At the physical layer, the signal-channel structure of the backscatter RFID channel, with a query-fading-signaling-fading structure, is radically different from conventional one-way wireless channels.
Measurements in \cite{Kim2003} and \cite{Griffin2009} showed that the backscatter RFID channel can be modeled as a two-way channel with a forward sub-channel and a backscattering sub-channel, and both sub-channels can be modeled as certain fading, depending on the radio propagation environment.
This two-way channel fades deeper than the conventional one-way channel and degrades the data transmission reliability and reading range, which are two important performance metrics in RFID systems.

Many efforts have been made on improving the performance of the backscatter RFID \cite{Ingram2001, Griffin2008, Griffin2009, Langwieser2010, DoYunKim2010, Denicke2012, Trotter2012, He 2011, He2012, He2013, Boyer2013, Karthaus2003, Nikitin2005, Fuschini2008, Xi2009, Bletsas2010, Chakraborty2011, Thomas2012, Kimionis2012, Boyer2012, Arnitz2013, Griffin2009B}.
Among those efforts, using multiple antennas for both tags and readers appears to be one practical and promising way. Such multiple-input multiple-output (MIMO) systems had a great success in conventional wireless communications \cite{Tarokh1998, Tarokh1999, Sandhu2000, Zheng2003, Tse2005} and were also investigated and found promising in RFID \cite{Ingram2001, Griffin2008, Griffin2009, Langwieser2010, DoYunKim2010, Denicke2012, Trotter2012}. A general MIMO backscatter RFID channel has $M$ query antennas on the reader, $L$ tag antennas on the tag and $N$ receiving antennas on the reader, as shown in Fig. \ref{Fig: MIMORFIDBigPic}. This MIMO setting can create spatial diversity and thus can improve the bit error rate (BER) performance and reading range of backscatter RFID. In \cite{Ingram2001}, simulations showed that with the MIMO setting, the range of backscatter RFID can be extended by a
factor of four or more in the pure diversity configuration and
that capacity can be increased by a factor of ten
or more in the spatial multiplexing configuration. In \cite{Griffin2008}, it was shown that backscatter diversity can mitigate the fading by changing the shape of the fading distribution which, along with
the increased RF tag scattering aperture, can result in a $10$ dB
gain at a BER of $10^{-4}$ and thus can lead to increased backscatter radio communication reliability and range (e.g., up to a $78$ percent range increase), which is consistent with a later result in \cite{DoYunKim2010}.
%More recently, \cite{Trotter2012} described how to overcome the extra path loss that RFID tags and RFID-enabled sensors experience at microwave frequencies as compared to UHF frequencies.
Except diversity gain, in \cite{Trotter2012} it was shown that additional antenna gains can be realized to mitigate or overcome extra path loss by using multiple antennas for narrowband signals centered at $5.8$ GHz.
The radio measurements of backscatter RFID with MIMO settings have also been investigated: in \cite{Griffin2009} the measurement was conducted at $5.8$ GHz, and experiment showed that diversity gains
are available for multiple-antenna RF tags and the results matched
well with the gains predicted using the analytic fading distributions
derived in \cite{Griffin2008}. In \cite{Denicke2012}, a method for the
determination of the channel coefficients between all antennas was
presented.
Another interesting research was conducted in \cite{Langwieser2010}, where researchers described a developed analog frontend for an RFID rapid prototyping system which allows for various real-time experiments to investigate MIMO techniques.

\subsection{Related Work}
The spatial diversity brought by MIMO settings for the backscatter RFID has been analytically studied recently.
With the quasi-static fading assumption, it was shown that for the $M \times L \times N$ backscatter RFID channel, the diversity order achieves $\min(N,L)$ for the uncoded case \cite{He2012}, and the diversity order achieves $L$ for the orthogonal space-time coded case \cite{He2013} \cite{Boyer2013}. Moreover, the diversity order cannot be greater than $L$ \cite{Boyer2013}.
%
%The research based on using multi-antenna to create spatial diversity and hence obtain diversity gain, while all this kind of research is based on uniform query, for which the query end sends out the same query signal from its $M$ query antennas over all symbol times. This performance limit can be
All the above studies that use MIMO settings to exploit the diversity gain for the backscatter RFID are based on the uniform query, for which the query antennas send the same signal over all symbol times. Since \cite{Ingram2001, Griffin2008}, where the $M \times L \times N$ backscatter RFID channel was formulated, there is no other query signaling methods have been considered. This is because the previous understanding is that, since spatial diversity can only be obtained by duplicating the information and transmitting it over multiple branches, while the query end is not the information source, designs of the query signal cannot bring spatial diversity in quasi-static channels. In this paper, however, we show that in quasi-static channels, the query signals can create \emph{time diversity} via multiple query antennas and thus improve the performance for the backscatter RFID \emph{significantly}. Our result does not follow the achievable diversity order of the $M \times L \times N$ channel reported in previous findings for the uniform query. We also analytically study the performance of the proposed unitary query.  Due to the specific signaling and fading structure of the backscatter RFID channel, the pairwise error probability (PEP) and even the diversity order are not trackable for the unitary query, we thus provide a new measure which can compare the PEP performance of the proposed unitary query with that of the conventional uniform query. We summarize the major contribution of this paper in the following.

\subsection{Contributions}
The major contributions of this work include:
\begin{itemize}
%  \item We provide a re-understanding of the role of the query end in the $M \times L \times N$ channel, and we show that multiple query antennas can transform their spatial advantage into \emph{time diversity}. Thus the signal design for query antennas has a great potential to leverage the achievable diversity order of the $M \times L \times N$ channel.
  \item We propose \emph{unitary query} which is performed at the reader query end. The proposed unitary query can create time diversity within each channel coherent interval, and thus can improve the BER performance significantly. Compared with case that the conventional uniform query is employed at the query end and space-time coding is employed at the tag end, the case that the proposed unitary query is employed at the query end and the same space-time code is employed at the tag end can yield a much better performance, which can be as large as $5-10$ dB in mid SNR regimes.
  \item We analytically study the performance of the proposed unitary query.  Due to the specific signaling and fading structure of the $M \times L \times N$ backscatter channel, the PEP and even the diversity order (i.e. the conventional measure) are not trackable for the unitary query, we thus derived a new measure which can compare the PEP performance of the unitary query with that of the uniform query. The derived measure in this paper can be used as criteria for designing the MIMO backscatter RFID system.
  \item We present analysis to show that the proposed \emph{unitary query} has a very practical meaning: for conventional uniform query, to improve the performance \emph{significantly}, equipping multiple antennas on the tag is a must. By contrast, for the proposed unitary query, to improve the performance significantly, equipping multiple antennas on the tag is \emph{not a must}. The proposed unitary can transfer the complexity requirements from the tag to the reader, and allows single-antenna tag to have high performance.
\end{itemize}

This paper is organized as follows: We give a brief introduction of the MIMO backscatter RFID channel in Section \ref{Sec: Channel_Model}. We propose the unitary query in Section \ref{Sec: Unitary_Querry}, and derive a new measure for the performance of the unitary query. In Section \ref{Sec: Examples and Simulations}, we study a few examples and conduct the corresponding simulations. Finally we summarize our work in Section \ref{Sec: Conclusion}.

\emph{Notations}: In this paper,  $\mathbb{Q}(\cdot)$ means the $Q$ function; $\mathbb{P}(\cdot)$, $\mathbb{E}_X(\cdot)$, $X|Y$, $\|\cdot\|_F$, $rank(\cdot)$, $\|\cdot\|$, $(\cdot)^T$, and $(\cdot)^H$ denote the probability of an event, the expectation over the density of $X$, the conditional random variable of $X$ given $Y$, the Frobenius norm of a matrix, the rank of a matrix, the magnitude of a complex number, the transpose, and the conjugate transpose, respectively; $X \sim Y$ means that $X$ is identically distributed with $Y$.

%At the reader query end, we propose a novel scheme called \textit{unitary query}. To our best knowledge, it is the first time that the unitary query has been proposed in RFID. In previous literature for MIMO backscatter RFID channels, only the uniform query was considered, and the understanding of query signals was that they only play a role as an energy provider for the RFID tag and thus cannot provide spatial diversity. In this chapter, however, we show that in quasi-static channels, the query signals can provide time diversity via multiple reader query antennas for some space-time codes, and hence improve the performance for the backscatter RFID significantly. The unitary query can bring time diversity within each channel coherent interval, and thus can improve the performance of backscatter RFID significantly.
%
%We also analytically study the performance of the proposed unitary query.  Due to the specific signaling and fading structure of the backscatter RFID channel, the PEP and even the diversity order are not trackable for the unitary query, we thus provide a new measure which can compare the PEP performance of the unitary query with that of the uniform query.

\section{The $M \times L \times N$ MIMO Backscatter RFID Channel}\label{Sec: Channel_Model}
The backscatter RFID has three operational ends: the reader query end (i.e., the set of reader transmitting antennas), the tag end (i.e., the set of tag antennas), and the reader receiver end (i.e., the set of reader receiving antennas). These three ends can be mathematically modeled by an $M\times L\times N$ dyadic backscatter channel which consists of $M$ reader transmitter antennas, $L$ RF tag antennas, and $N$ reader receiver antennas \cite{Ingram2001, Griffin2008, DoYunKim2010, He2013, Boyer2013}, as shown in Fig. \ref{Fig: MIMORFIDBigPic}.
%The forward channel $h_{ml}^f$ represents the
%propagation path from the $m$-th reader transmitter to the $l$-th RF
%tag antenna, while the backscatter channel $h_{ln}^b$ represents the path in
%which the carrier signal is reflected by the $l$-th tag antenna to
%the $n$-th reader receiver. The forward and backscatter links
%that terminate or originate at the same tag antenna can be
%correlated, as indicated in Fig. \ref{Fig: RFID_MIMO_General}, where $\rho_{ml}^{ln}$
%denotes the link correlation coefficient between the forward link
%$h_{ml}^f$ and the backscatter link $h_{ln}^b$.
In a quasi-static wireless channel, this MIMO structure can be summarized by using the following matrices: More specifically,
\begin{align}\label{Eq: MatrixQ}
\mathbf{Q}=\left(
           \begin{array}{ccc}
             q_{1,1}    & \cdots        &  q_{1,M}   \\
             \vdots     & \ddots        & \vdots     \\
             q_{T,1}    & \cdots        & q_{T,M}    \\
           \end{array}
         \right)
\end{align}
is the query matrix (with size $T \times M$), representing the query signals sent from the the $M$ reader query (transmitting) antennas to the tag over $T$ time slots (i.e. $T$ symbol times);
\begin{align}\label{Eq: MatrixH}
\mathbf{H}=\left(
           \begin{array}{ccc}
             h_{1,1}    & \cdots        & h_{1,L}   \\
             \vdots     & \ddots        & \vdots     \\
             h_{M,1}    & \cdots        & h_{M,L}    \\
           \end{array}
         \right)
\end{align}
is the channel gain matrix (with size $M \times L$) from the reader transmitter to the tag, representing the forward sub-channels;
\begin{align}\label{Eq: MatrixC}
\mathbf{C}=\left(
           \begin{array}{ccc}
             c_{1,1}    & \cdots        & c_{1,L}   \\
             \vdots     & \ddots        & \vdots     \\
             c_{T,1}    & \cdots        & c_{T,L}    \\
           \end{array}
         \right)
\end{align}
is the coding matrix (with size $T \times L$), where the tag transmits space-time coded or uncoded symbols from its $L$ antennas over $T$ time slots; and
\begin{align}\label{Eq: MatrixG}
\mathbf{G}=\left(
           \begin{array}{ccc}
             g_{1,1}    & \cdots  & g_{1,N}    \\
             \vdots     & \ddots  &  \vdots     \\
             g_{L,1}    & \cdots  & g_{L,N} \\
           \end{array}
         \right),
\end{align}
is the channel gain matrix (with size $L \times N$) from the tag to the reader receiver, representing the backscattering sub-channels.
Finally the received signals at $N$ reader receiving antennas over $T$ time slots are represented by matrix $\mathbf{R}$ with size $T \times N$:
\begin{align}\label{Eq: RFID_Channel_Model}
\mathbf{R}=((\mathbf{Q}\mathbf{H})\circ\mathbf{C})\mathbf{G}+\mathbf{W}
\end{align}
where $\circ$ is the Hadamard product, and the matrix $\mathbf{W}$ is with the same size as that of $\mathbf{R}$, representing the noise at the $N$ reader receiving antennas over $T$ time slots. Typically, both $\mathbf{H}$ and $\mathbf{G}$ are modeled as full rank matrices with i.i.d complex Gaussian entries, and $\mathbf{W}$ is AWGN.
%When compared with the conventional one-way MIMO wireless channel:
%\begin{align}\label{Eq: One_Way_Channel_Model}
%\mathbf{R}=\mathbf{C}\mathbf{G}+\mathbf{W},
%\end{align}
%the backscatter structure in \eqref{Eq: RFID_Channel_Model}

The signal-channel structure of the $M \times L \times N$ RFID channel is radically different from conventional wireless channels, and can be characterized as a \emph{query-fading-coding-fading} structure.
Compared with the conventional one-way wireless channel, this signal-channel structure not only has one more layer of fading $\mathbf{H}$ but also one more signaling mechanism represented by the query matrix $\mathbf{Q}$. In addition, the backscatter principle makes the received signals not a simple series of linear transformations of transmitted signals and channel gains, but actually there involves a non-linear structure in the backscatter RFID channel, which is the result from the Hadamard product in \eqref{Eq: RFID_Channel_Model}. Because it has such a special signaling-channel structure, the backscatter RFID channel behaves completely different from that of the one-way channel \cite{He2013, Boyer2013}. It is worth mentioning here that the keyhole channel also has two layers of fading, however, the keyhole channel and the backscatter RFID channel are essentially different. The keyhole channel is still a one-way channel, as the signals sent out will not be reflected back. In addition, the keyhole channel has only two operational ends (the transmitter and the receiver), while the backscatter channel has three operational ends and the information to be transmitted is at the middle end (the tag end). The essential differences of the two channels have been discussed in \cite{Boyer2013, He2013}, especially there is a detailed discussion in \cite{Boyer2013}. In general, the $M \times L \times N$ backscatter RFID channel is more complicated than the keyhole channel.
%In this chapter, we concentrate on the simplest query scheme and tag-signaling scheme, and show that, even for the simplest case, the MIMO backscatter RFID channel has interesting properties. In the next chapter, we will investigate MIMO backscatter RFID channels under more generalized query and signaling cases.

\begin{figure}
\centering
  \includegraphics[scale=0.6]{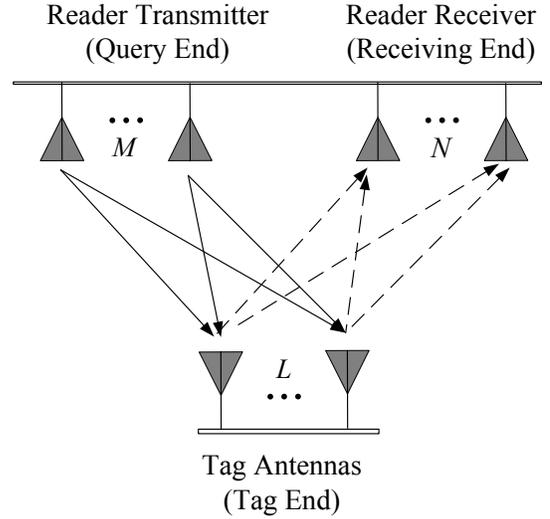}\\
  \caption{The $M \times L \times N$ backscatter RFID channel. The channel consists three operational ends: the query end (with $M$ query antennas), the tag end (with $L$ tag antennas) and the receiving end (with $N$ receiving antennas). The query antennas transmit  unmodulated (query) signals to the RF tag and the RF tag scatters a modulated signal back to the reader.}\label{Fig: MIMORFIDBigPic}
\end{figure}

\section{Unitary Query for Backscatter RFID}\label{Sec: Unitary_Querry}
Recall that in the backscatter RFID channel, there are three operational ends: the query end, the tag end, and the receiving end. In the previous literature, the understanding of the query end was that the design of query signals can not improve the BER performance. This is based on the following explanation: In the quasi-static channel, since spatial diversity can only be created by duplicating the information and transmitting it from multiple branches, while the query end is not the information source, it plays only as a role of energy provider for the tag when the tag is transmitting its information. However, in this section we reconsider the previous understanding and propose the \textit{unitary query}. We show that the proposed unitary query can improve the performance of space-time code (STC), and this improvement can be significant. Unitary query improves the BER performance by transforming the spatial advantage to the time diversity via multiple antennas at the reader query end.

In this paper, it is assumed that the fading channel is quasi-static, i.e., the channel is constant over a long period of time and changes in an independent manner. This quasi-static assumption is valid as long as the transmitter and the receiver is not moving in high velocity, and it is one of the major assumptions for many wireless communication systems including many RFID systems \cite{He2013} \cite{Boyer2013}. Since the channel does not change within the coherent time, only spatial diversity can be provided within the coherent time for the conventional one-way channel. In the backscatter RFID channel, when the conventional uniform query is used (i.e.,
all the $M$ query antennas send out the same signals over $T$ time slots), the query matrix is given by
\begin{align}
 \mathbf{Q}_{\textmd{uniform}}
 =\frac{1}{\sqrt{M}}\left(
           \begin{array}{ccc}
             1    & \cdots        & 1   \\
             \vdots     & \ddots        & \vdots     \\
             1    & \cdots        & 1    \\
           \end{array}
         \right).
\end{align}
The uniform query is used as the query method for all previous studies of the $M \times N \times L$ backscatter channel, and no further investigation has been made on the query signal design since the $M \times N \times L$ backscatter channel has been formulated in \cite{Ingram2001, Griffin2008}.
The reason that no other query signal design method has been considered probably from the understanding that the spatial diversity from the transmitter can only be made when transmitting duplicated information from different antennas, and since the query signals do not carry information, the spatial diversity can only be made from the tag antennas and the reader receiving antennas. %Therefore there was no consideration about the signal design at the query end, and only uniform query has been used.
%Due to the difficulty of obtaining the asymptotic PEP and the even diversity order for the proposed unitary query, we also provide a new measure for performance analysis. With the new measure, we do not need to exactly calculate the PEP but can still compare performances of different query and space-time coding schemes.
%Recall that in Chapter \ref{Ch: Identical}, the channel model of the $M \times L \times N$ backscatter RFID can be characterized by
%\begin{align}\label{Eq: RFID_Channel_Model_In_Chapter3}
%\mathbf{R}=\mathbf{Q}\mathbf{H}\circ\mathbf{C}\mathbf{G}+\mathbf{W},
%\end{align}
%where both the forward sub-channels (represented by $\mathbf{H}$) and the backscattering sub-channels (represented by $\mathbf{G}$) are modeled as i.i.d. complex Gaussian random variables with zero mean and unity variance.

However, in general, query signals can be designed to follow any arbitrary $\mathbf{Q}$. In this paper, we propose the so-called unitary query, which satisfies the unitary condition:
\begin{align}
\mathbf{Q}_{\textmd{unitary}}\mathbf{Q}_{\textmd{unitary}}^H=\mathbf{I}.
\end{align}
Note that to satisfy the unitary condition we must have $T=M$, while, as long as there are at least  $M$ symbol times during the transmission period, we can always cast the query signals into blocks each of which has $T=M$ symbol times, and obtain the unitary query.
Since the above query matrix is unitary and the entries of $\mathbf{H}$ are i.i.d complex Gaussian, we have
\begin{align}\label{Eq: Matrix_X}
 \mathbf{Q_{\textmd{unitary}}}\mathbf{H} \sim \mathbf{X}
 =\left(
           \begin{array}{ccc}
             x_{1,1}    & \cdots        & x_{1,L}   \\
             \vdots     & \ddots        & \vdots     \\
             x_{T,1}    & \cdots        & x_{T,L}    \\
           \end{array}
         \right).
\end{align}
The resulting matrix $\mathbf{X}$ (with size $T\times L$) has i.i.d complex Gaussian entries $x_{t,l}$'s, so the unitary query actually transforms the forward channel $\mathbf{H}$, which is invariant over the $T$ time slots, into a channel $\mathbf{X}$ which varies over the $T$ time slots. We will show later that this variation over the $T$ time slots is the fundamental reason that the unitary query can bring additional time diversity and significant performance improvement for some STCs in the backscatter RFID channel.
When compared with that of the uniform query
\begin{align}\label{Eq: Matrix_Y}
 \mathbf{Q_{\textmd{uniform}}}\mathbf{H} \sim \mathbf{Y}
 =\left(
           \begin{array}{ccc}
             y_{1}    & \cdots        & y_{L}   \\
             \vdots     & \ddots        & \vdots     \\
             y_{1}    & \cdots        & y_{L}    \\
           \end{array}
         \right),
\end{align}
where $y_{l}$'s are i.i.d complex Gaussian.
Clearly the resulting matrix $\mathbf{Y}$ (also with size $T \times L$) has identical rows. Thus the uniform query transforms the full rank matrix $\mathbf{H}$ into a rank-one matrix, while the unitary query transforms $\mathbf{H}$ into another full rank matrix. $\mathbf{X}$ varies both temporally and spatially, while  $\mathbf{Y}$ only varies spatially.
In the following sub-section, we give a brief interpretation of the diversity of the proposed unitary query.
%The PEP performance can be obtained by
%\begin{align} \label{Eq: Qeury_PEP_Conditional}
%\textmd{PEP}(\bar{\gamma}) =\mathbb{E}\left(\mathbb{Q}\left(\sqrt{\frac{\bar{\gamma}}{M\times N}Z}\right)\right),
%\end{align}
%averaging over the densities of $H$ and $G$,

%\subsection{Decoding}
%Channel estimation
%To decode the
%$(\mathbf{Q}\mathbf{H})_{1,1}g$
\subsection{Interpretation for Unitary Query: Time Diversity within Coherent Interval}\label{Sec: Interpretation}
In the quasi-static channel, where the channel is highly correlated across consecutive symbols, no time diversity can be provided within one coherent time interval for the one-way channel, and time diversity can only be provided by interleaving symbols in different coherent time intervals. This also applies to the backscatter RFID channel when the conventional uniform query is employed.
The unitary query, however, utilizes the multiple query antennas, to create time diversity within channel coherent time.  Fig. \ref{Fig: TimeDiversity} shows that, with the conventional uniform query, the backscatter RFID channel still behaves like a quasi-static channel: the channel changes every $T$ symbol times; by contrast, when the unitary query is employed, the channel changes every $1$ symbol time.
%Diversity branches. A deep fade will wipe out all the codewords in the entire coherent time interval. While since within one coherent time interval, the risk of the is much lower than, the codewords can still be recovered.

An alternative interpretation based on geometry is shown in Fig. \ref{Fig: Geometry}. We consider the codewords $(c_{1,1}, c_{2,1}, \cdots, c_{T,1})$, which can be viewed as a point in a $T$-dimensional space. We can see that when the uniform query is applied, possible locations of the point $(c_{1,l},c_{2,l}, \cdots, c_{T,1})$ can only be mapped to the points on a straight line, which is only $1$-dimensional. However, when the unitary query is applied, possible locations of the point $(c_{1,l},c_{2,l}, \cdots, c_{T,1})$ can be mapped to any points in the entire $T$-dimensional space. From Fig. \ref{Fig: Geometry}, it is clear that this kind of full-dimensional spreading out of possible locations of the codewords by the unitary query may yield significant performance improvements.

%can only be mapped on a strait line by $\mathbf{Y}$ given in xxx
%can only lie on a strait line for any value that the forward channel can take;
%however, when unitary query is applied, the codewords pair can be mapped to any point in the 2-dimensional plan by $\mathbf{X}$ given in xxx. And it is very clear from Fig. xxx that, this kind of full-dimensional spreading out of codewords by unitary query may bring significant performance improvement.
%possible codewords associated with channel gains are separate in all dimensions and there are many of them spread out the space.
\begin{figure}
\centering
  \includegraphics[scale=0.66]{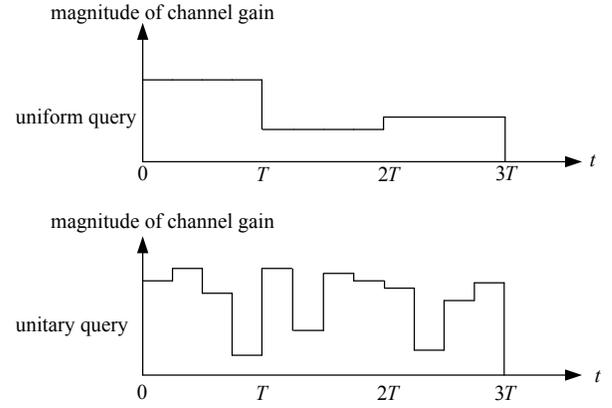}\\
  \caption{The proposed unitary query can create time diversity within channel coherent time for the $M \times L \times N$ backscatter channel. By employing the unitary query, the channel is independent for each symbol time in ideal situations, and thus the risk of having all codewords in the entire coherent time being wiped out decreases. This type of time diversity within coherent time does not exist in the conventional one-way channel.}\label{Fig: TimeDiversity}
\end{figure}

\begin{figure}
\centering
  \includegraphics[scale=0.66]{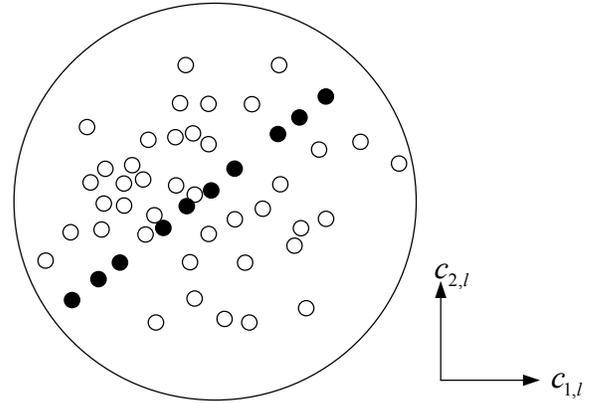}\\
  \caption{The proposed unitary query can map possible locations of codewords (as a point in $T$ dimensional space) within a coherent time interval into any points in a $T$-dimensional space (small white circles), while the conventional uniform query can only map possible locations of codewords within a coherent interval into a $1$-dimensional space (small black circles).}\label{Fig: Geometry}
\end{figure}

\subsection{New Performance Measure for the $M \times N \times L$ Channel}\label{Sec: New_Measure}
Now we need to study the performance of the unitary query. In previous literature, the performance when the $M \times N \times L$ channel employs the uniform query was investigated in \cite{Griffin2008, He2012, He2013, Boyer2013}, and it was shown that, the analysis is very difficult even with the conventional uniform query. With the proposed unitary query being employed, the analysis will be more difficulty, the diversity order is not trackable. We thus derive a new measure other than the conventional diversity order for the performance analysis. The new measure is based on the ranks of certain carefully constructed random matrices.

When the $M \times N \times L$ channel employs the unitary query at the query end, and employs space-time coding at the tag end, it has an equivalent channel model as
\begin{align}\label{Eq: RFID_Channel_Model_Uniform_Query}
\mathbf{R}=(\mathbf{X}\circ\mathbf{C})\mathbf{G}+\mathbf{W},
\end{align}
where $\mathbf{X}$ is given in \eqref{Eq: Matrix_X}.
When the $M \times N \times L$ channel applies uniform query at the query end, and space-time coding at the tag end, it has an equivalent channel model as
\begin{align}\label{Eq: RFID_Channel_Model_Uniform_Query}
\mathbf{R}=(\mathbf{Y}\circ\mathbf{C})\mathbf{G}+\mathbf{W},
\end{align}
where where $\mathbf{Y}$ is given in \eqref{Eq: Matrix_Y}.
Now we define the codewords difference matrix for codewords matrices $\mathbf{C}$ and $\mathbf{C}'$ as
\begin{align}
\Delta=\mathbf{C}-\mathbf{C}'=\left(
           \begin{array}{ccc}
             \delta_{1,1}    & \cdots        & \delta_{1,T}   \\
             \vdots     & \ddots        & \vdots     \\
             \delta_{L,1}    & \cdots        & \delta_{L,T}    \\
           \end{array}
         \right).
\end{align}
The PEP is the probability that the receiver decide erroneously in favor of the codewords matrix $\mathbf{C}'$ a when the $\mathbf{C}$ is actually transmitted, for unitary query, the PEP is can be evaluated as
\begin{align} \label{Eq: Unitary_Query_PEP}
 \textmd{PEP}_{\textmd{X}}(\bar{\gamma})
 =\mathbb{E}_{\mathbf{H},\mathbf{G}}\left(\mathbb{Q}\left(\sqrt{\bar{\gamma}Z_X/2}\right)\right),
\end{align}
where
\begin{align}\label{Eq: CW_Distance_Unitary_Querry}
Z_X&=\|(\mathbf{X}\circ\mathbf{C})\mathbf{G}-(\mathbf{X}\circ\mathbf{C'})\mathbf{G}\|^2_F \nonumber\\
   &=\|(\mathbf{X}\circ\Delta)\mathbf{G}\|^2_F,
\end{align}
is the random variable which represents the squared distance between the codewords matrices $\mathbf{C}$ and $\mathbf{C}'$ when unitary query is employed and the tag uses space-time coding.
Similarly, for the uniform query, the PEP is given by
\begin{align} \label{Eq: Uniform_Query_PEP}
 \textmd{PEP}_{\textmd{Y}}(\bar{\gamma})
 =\mathbb{E}_{\mathbf{H},\mathbf{G}}\left(\mathbb{Q}\left(\sqrt{\bar{\gamma}Z_Y/2}\right)\right),
\end{align}
where
\begin{align}\label{Eq: CW_Distance_Uniform_Querry}
Z_Y&=\|(\mathbf{Y}\circ\mathbf{C})\mathbf{G}-(\mathbf{Y}\circ\mathbf{C'})\mathbf{G}\|^2_F \nonumber\\
   &=\|(\mathbf{Y}\circ\Delta)\mathbf{G}\|^2_F
\end{align}
is the random variable which represents the squared distance between the codewords matrices $\mathbf{C}$ and $\mathbf{C}'$ when uniform query is employed and the tag uses space-time coding. $\bar{\gamma}$ in the above equations is the averaged signal-to-noise ratio (SNR).

%For one-way channel, which have a simpler signaling and fading structure than that of the $M \times N \times L$ channel, the asymptotic PEP and diversity order can be derived in closed form, the code design criteria can be derived accordingly.
Quite different from that of the one-way channel, directly evaluating the PEPs in \eqref{Eq: Unitary_Query_PEP} and \eqref{Eq: Uniform_Query_PEP} is not feasible  because the distributions of $Z_X$ and $Z_Y$ are not trackable when general space-time code is considered at the tag end.
Even for the case when the uniform query is employed (corresponds to the distribution of $Z_Y$), the asymptotic PEP can only be obtained for two special coding cases: the orthogonal space-time code\cite{He2013, Boyer2013} and the uncoded case \cite{He2012}. When the proposed unitary query is employed (corresponds to the distribution of $Z_X$), evaluating the PEP will be even harder. In this paper, we reconsider the evaluation of PEP and provide a new measure for the PEP performance for the $M \times L \times N$ channel, to overcome the above difficulties. This new measure can provide a deep understanding of the performance of the channel, and can be used to compare the performances between the unitary query and the uniform query.
Instead of considering the squared codewords  distance as a whole, we treat it in a time fashion.  When the unitary query is employed, at time $t$, the squared codewords distance is given by
\begin{align}
Z_X^t & =  \|(x_{t,1},\cdots,x_{t,L})\circ (\delta_{1,t},\cdots,\delta_{L,t})  \mathbf{G} \|_F^2 \nonumber\\
    & =  \|(x_{t,1},\cdots,x_{t,L})   \Delta_t  \mathbf{G} \|_F^2 \nonumber\\
\end{align}
where $\Delta_t$ is defined as
\begin{align}
\Delta_{t}\triangleq \left(
               \begin{array}{ccc}
                 \delta_{1,t} &        &              \\
                              & \ddots &              \\
                              &        & \delta_{L,t} \\
               \end{array}
             \right),
\end{align}
then over the $T$ time slots we have
\begin{align}
Z_X&=\sum_{t=1}^T \|(x_{t,1},\cdots,x_{t,L})   \Delta_t  \mathbf{G} \|_F^2 \nonumber\\
 &=\sum_{t=1}^T \|(x_{t,1},\cdots,x_{t,L})   \mathbf{E}_t \|_F^2,
\end{align}
where $\mathbf{E}_t$ is defined as
\begin{align}\label{Eq: MatrixEt}
\mathbf{E}_t \triangleq \Delta_t  \mathbf{G}.
\end{align}
We will see later that the ranks of the carefully constructed random matrices $\mathbf{E}_t$'s determine the performance for the unitary query.

When the uniform query is employed, the squared codewords distance at time $t$ is given by
\begin{align}
Z_Y^t & =  \|(y_{1},\cdots,y_{L})\circ (\delta_{1,t},\cdots,\delta_{L,t})  \mathbf{G} \|_F^2 \nonumber\\
    & =  \|(y_{1},\cdots,y_{L})   \Delta_t  \mathbf{G} \|_F^2 \nonumber\\
\end{align}
and over the $T$ time slots we have
\begin{align}
Z_Y & = \sum_{t=1}^T \|(y_{1},\cdots,y_{L})   \mathbf{E}_t \|_F^2 \nonumber\\
    &  = \|(y_{1},\cdots,y_{L})   (\mathbf{E}_1, \cdots, \mathbf{E}_T)\|_F^2.
\end{align}
Note that inside a $\|\cdot\|_F$ operator, the columns of the matrix $(\mathbf{E}_1, \cdots, \mathbf{E}_T)$ are interchangeable, therefore we have
\begin{align}
Z_Y &   = \|(y_{1},\cdots,y_{L})   (\mathbf{D}_1, \cdots, \mathbf{D}_N)\|_F^2,
\end{align}
where $\mathbf{D}_n$'s are defined as
\begin{align}
\mathbf{D}_n \triangleq \Delta  \mathbf{G}_n,
\end{align}
and $\mathbf{G}_n$'s are defined as
\begin{align}
\mathbf{G}_n\triangleq \left(
               \begin{array}{ccc}
                 g_{1,n} &        &              \\
                              & \ddots &              \\
                              &        & g_{L,n} \\
               \end{array}
             \right),
\end{align}
for $n=1,\cdots,N$.
Also, we will see later that the rank of the carefully constructed random matrix
\begin{align}\label{Eq: MatrixD}
\mathbf{D}\triangleq (\mathbf{D}_1, \cdots, \mathbf{D}_N)
\end{align}
determines the performance for the uniform query.

%%%%%%%%%%%%%%%%%%%% Lemma 1 : Rank of E_t %%%%%%%%%%%%%%%%%%%%%%%%%%%%%%%%%%%
Now we give the following two Lemmas about the ranks of the random matrices $\mathbf{E}_t$'s and the rank of the matrix $\mathbf{D}$.
\begin{Lemma}\label{Lemma: Rank_of_Et}
For the matrices $\mathbf{E}_t$'s defined in \eqref{Eq: MatrixEt}, we have $\textmd{rank}(\mathbf{E}_t)=\min(N,L_t^*)$  with probability (w.p.) $1$ for all $t \in \{1, \cdots, T\}$, where $L_{t}^*$ is the number of non-zero elements of the $t$-th column of the codewords difference matrix $\Delta$.
\end{Lemma}
\begin{proof}[Proof of Lemma \ref{Lemma: Rank_of_Et}]
See the appendix.
\end{proof}

%%%%%%%%%%%%%%%%%%%% Lemma 2 %%%%%%%%%%%%%%%%%%%%%%%%%%%%%%%%%%%

\begin{Lemma}\label{Lemma: Rank_of_D}
For the matrix $\mathbf{D}$ defined in \eqref{Eq: MatrixD}, we have $\textmd{rank}(\mathbf{D}) = \min(N \times \textmd{rank}(\Delta),L)$  with probability $1$, where $L$ is the number of non-zero columns of the codewords difference matrix $\Delta$.
\end{Lemma}
\begin{proof}[Proof of Lemma \ref{Lemma: Rank_of_D}]
See the appendix.
\end{proof}

With understanding the ranks of the matrices $\mathbf{E}_t$'s and the rank of the matrix $\mathbf{D}$, we introduce the following theorem of the new measure for the unitary query and the uniform query.

%%%%%%%%%%%%%%%%%%%%%%%%%%%%%%%%%%%%%%%%%%%%%%%%%%%%%%%%%%%%%%%%%%%%%%%%%%%%
%%%%%%%%%%%%%%%%%%%%%% Theorem 1 %%%%%%%%%%%%%%%%%%%%%%%%%%%%%%%%%%%%%%%%%%%
%%%%%%%%%%%%%%%%%%%%%%%%%%%%%%%%%%%%%%%%%%%%%%%%%%%%%%%%%%%%%%%%%%%%%%%%%%%%
\begin{Theorem}\label{Th: New_Measure}
In asymptotic high SNR regimes, the PEP performances of space-time codes with the unitary query and the uniform query in the $M \times N \times L$ backscatter RFID channel given in \eqref{Eq: RFID_Channel_Model} can be measured by
\begin{align}
R_{unitary}=\sum_{t=1}^T\min(N,L_t^*),
\end{align}
and
\begin{align}
R_{uniform}=\min(N \times \textmd{rank}(\Delta),L),
\end{align}
respectively, where $L_{t}^*$ is the number of non-zero elements of the $t$-th column of the codewords difference matrix $\Delta$. In other words, if
\begin{align}
R_{unitary}>R_{uniform},
\end{align}
we have
%\begin{align}\label{Eq: Weaker_Measure}
%\lim_{\bar{\gamma}\rightarrow \infty}
%\frac{\textmd{PEP}_{Z_X}(\bar{\gamma})}{\text{PEP}_{Z_Y}(\bar{\gamma})}<1.
%\end{align}
%Further, if  $\mathbb{E}_{\mathbf{G}}\left(\prod_{t=1}^T\prod_{i=1}^{\min(N,L_t^*)}\frac{1}{\lambda_{i,t}}\right)$ and $\mathbb{E}_{\mathbf{G}}\left(\prod_{i=1}^{R_{\textmd{uniform}}}\frac{1}{\lambda_{i}^*}\right)$ are finite, we have a stronger version as
\begin{align}
\lim_{\bar{\gamma}\rightarrow \infty}\label{Eq: Stronger_Measure}
\frac{\textmd{PEP}_{Z_X}(\bar{\gamma})}{\text{PEP}_{Z_Y}(\bar{\gamma})}\rightarrow 0;
\end{align}
if
\begin{align}
R_{unitary}<R_{uniform},
\end{align}
we have
\begin{align}
\lim_{\bar{\gamma}\rightarrow \infty}\label{Eq: Stronger_Measure}
\frac{\textmd{PEP}_{Z_Y}(\bar{\gamma})}{\text{PEP}_{Z_X}(\bar{\gamma})}\rightarrow 0;
\end{align}
and if
\begin{align}
R_{unitary}=R_{uniform},
\end{align}
we have
\begin{align}
\lim_{\bar{\gamma}\rightarrow \infty}\label{Eq: Stronger_Measure}
\frac{\textmd{PEP}_{Z_X}(\bar{\gamma})}{\text{PEP}_{Z_Y}(\bar{\gamma})}=c>0;
\end{align}
where $c$ is some positive constant.
\end{Theorem}
\begin{proof}[Proof of Theorem \ref{Th: New_Measure}]
See the appendix.
\end{proof}

The new measure in Theorem \ref{Th: New_Measure} can be used to compare the PEP performances of the unitary query and the uniform query in \emph{large scale} (i.e., if one measure is large than the other, its performance will be much better than that of the other). Therefore in some sense the new measure is similar to the diversity order, but not exactly the same. We give a brief discussion of three possible cases.\\
\emph{Case 1}:$R_{\textmd{unitary}}>R_{\textmd{uniform}}$\\
In this case, the performance of the unitary query will be much better than that of the uniform query. Most well designed space-time codes fall into this case, and can drive the full potential of the $M \times L \times N$ backscatter RFID channel.\\
\emph{Case 2}:$R_{\textmd{unitary}}<R_{\textmd{uniform}}$ \\
In this case, the performance of the uniform query will be much better than that of the uniform query. However, only rare space-time codes fall into this case. Such space-time codes cannot drive the full potential of the $M \times L \times N$ backscatter RFID channel and thus are not preferred in the $M \times L \times N$ channel.\\
\emph{Case 3}:$R_{\textmd{unitary}}=R_{\textmd{uniform}}$ \\
In this case, the performance of the unitary query will be similar as that of the uniform query, while the unitary query will still outperform the uniform query, though the improvement will not be significant.

Note that in the above three cases, \emph{Case 1} can achieve the full potential of the $M \times L \times N$ channel, and is usually preferred.

%%%%%%%%%%%%%%%%%%%%%%%%%%%%%%%%%%%%%%%%%%%%%%
%%%%%%% figure unitary query PEP %%%%%%%%%%%
%%%%%%%%%%%%%%%%%%%%%%%%%%%%%%%%%%%%%%%%%%%%%%
\section{Performance Evaluations}\label{Sec: Examples and Simulations}
In this section, we give a few examples and provide corresponding simulation results for the proposed unitary query and the conventional uniform query. We will see by how much the unitary query can improve the performance and how the unitary query transfers the complexity requirement from the tag end to the reader end for high performance systems. In the following simulations, we use the same channel model as in previous real measurements \cite{Kim2003} \cite{Griffin2009} and analytical studies \cite{Griffin2008, He2012, He2013, Boyer2013} of the $M \times N \times L$ backscatter RFID channel. More specifically, the entries of both $\mathbf{H}$ in \eqref{Eq: MatrixH} and those of $\mathbf{G}$ in \eqref{Eq: MatrixG} follow i.i.d complex Gaussian distribution, with zero mean and unit variance, and the fading is quasi-static. Given a codewords difference matrix $\Delta$, the new performance measure ($R_{uniform}$) for the conventional uniform query  given in Theorem \ref{Th: New_Measure} is based on the rank of the random matrix $\mathbf{D}$ defined in \eqref{Eq: MatrixD}, and the new performance measure ($R_{unitary}$) for the proposed unitary query is based on the ranks of the random matrices $\mathbf{E}_t$'s defined in \eqref{Eq: MatrixEt}.

\subsection{Tag with Multiple Antennas}
When the uniform query is employed, the limit of the performance is given by
\begin{align}\label{Eq: Ex1Uniform}
R_{uniform}=\min(N\times rank(\Delta), L)\leq L,
\end{align}
which means no matter how many antennas are equipped in the channel and whatever the space-time code is, the performance  has a bottleneck determined by $L$.
However, the unitary query can break through this bottleneck and bring a significant improvement:
\begin{align}\label{Eq: Ex1Unitary}
R_{unitary}=\sum_{t=1}^T\min(N,L_t^*)\le TL.
\end{align}
With some space-time codes, the above measure $R_{unitary}$ can achieve $TL$.
We give the following example to illustrate this and show how much gain the unitary query can bring.\\

\emph{\textbf{Example 1}}
Consider the $2 \times 2 \times 2$ backscatter RFID channel, i.e.
\begin{align}
\mathbf{H}=\left(
         \begin{array}{cc}
           h_{1,1}   & h_{1,2} \\
           h_{2,1}   & h_{2,2} \\
         \end{array}
       \right),
\mathbf{G}=\left(
         \begin{array}{cc}
           g_{1,1}   & g_{1,2} \\
           g_{2,1}   & g_{2,2} \\
         \end{array}
       \right),
\end{align}
where the entries of $\mathbf{H}$ and $\mathbf{G}$ are i.i.d. complex Gaussian with zero mean and unit variance,
and the following codewords difference matrix resulted from the space-time code employed at the tag end:
\begin{align}\label{Eq: CodeWordDifferenceMatrix}
\Delta=\left(
         \begin{array}{cc}
           1   & -2 \\
           1.5 & 2.5 \\
         \end{array}
       \right).
\end{align}
In this case $M=2$, $L=2$, $N=2$, $T=2$, $rank(\Delta)=2$, and $L_1^*=L_2^*=2$.
Based on Theorem \ref{Th: New_Measure}, when the unitary query is employed we have
\begin{align}
R_{unitary}&=\sum_{t=1}^T \min(N,L_t^*) \nonumber\\
           &=\min(2,2)+\min(2,2)=4,
\end{align}
and when the uniform query is employed we have
\begin{align}
R_{uniform}&=\min(N \times rank(\Delta),L) \nonumber\\
           &=\min(2 \times 2, 2)=2.
\end{align}
Therefore the performance of the unitary query is expected to be much better than that of the uniform query. Simulations confirm this as we can see in Fig. \ref{Fig: UnitaryVsUniformN2}: there is a significant gain by employing the unitary query for the $2 \times 2 \times 2$ backscatter channel. We observe a $5$ to $7$ dB gain in the SNR regimes of $10$ to $15$ dB when the system employs unitary query, and the gain increases as the SNR increases. This gain brought by the unitary query can be considered as the time diversity gain that has been illustrated in Section \ref{Sec: Interpretation}.\\

\emph{\textbf{Example 2}}
We consider a case that the tag end employs the same space-time code as that in \emph{Example 1} but a different antenna setting: a $2 \times 2 \times 1$ backscatter RFID channel, i.e.
\begin{align}
\mathbf{H}=\left(
         \begin{array}{cc}
           h_{1,1}   & h_{1,2} \\
           h_{2,1}   & h_{2,2} \\
         \end{array}
       \right),
\mathbf{G}=\left(
         \begin{array}{c}
           g_{1,1}\\
           g_{2,1}\\
         \end{array}
       \right).
\end{align}
In this case $M=2$, $L=2$, $N=1$, $T=2$, $rank(\Delta)=2$, $L_1^*=L_2^*=2$, and the measures are given by
\begin{align}
R_{unitary}&=\sum_{t=1}^T \min(N,L_t^*)\nonumber\\
           &=\min(1,2)+\min(1,2)=2,
\end{align}
and
\begin{align}
R_{uniform}&=\min(N \times rank(\Delta), L)\nonumber\\
           &=\min(1 \times 2, 2)=2.
\end{align}
Since $R_{unitary}=R_{uniform}$, by Theorem \ref{Th: New_Measure} the unitary query still outperforms the uniform query but the improvement is not significant, as shown in Fig. \ref{Fig: UnitaryVsUniformN1}.
In this case, with the given code difference matrix in (\ref{Eq: CodeWordDifferenceMatrix}), the $2 \times 2 \times 1$ channel achieves the full potential for the uniform query but does not achieve the full potential for the unitary query, that is reason why the unitary query outperforms the uniform query but the gain is not significant.

%Note that $R_{uniform}=rank(\Delta)$, therefore the $2 \times 2 \times 1$ drive the full potential of the space-time code given in xxx when employing uniform query. By contrast, the  $2 \times 2 \times 1$ cannot drive the full potential of the unitary query for the space-time code given in \label{Eq: CodeWordDifferenceMatrix}, the $$
%
%Given the codewords difference matrix given in xxx, the $2 \times 2 \times 1$ drives the full potential of the uniform query as $R_{uniform}$ in this case can be rewritten as $R_{uniform}=\min(1 \times rank(\Delta), L)=\min(rank(\Delta), rank(\Delta))=rank(\Delta)$.
%In this case, with the codewords difference matrix given in xxx, the uniform query reaches its maximum mwith the given space-time code, while the unitary

\subsection{Tag with Single Antenna}
In practice, since equipping multiple antennas on the tag increases the complexity and even the size of the tag, single-antenna tags are always preferred. However, with the conventional uniform query, the performance of the single-antenna tag ($L=1$) is quite limited. As we can see that when $L=1$, the performance measure for the conventional query uniform is
\begin{align}
R_{uniform}&=\min(N\times rank (\Delta), 1)=1.
\end{align}
It means that, when the conventional uniform query is employed, significant performance improvement can never be made for single-antenna tags. However, for the unitary query, the measure is given by
\begin{align}
R_{unitary}=\sum_{t=1}^T \min(N, L_t^*)\leq \sum_{t=1}^T \min(N, 1)=T.
\end{align}
Clearly when $L_t^*=L=1$ for all $t$, $R_{unitary}$ achieves $T$. Therefore, with the unitary query, carefully choosing coding scheme can lead to significant improvements for single-antenna tags. We use the following example to illustrate this.\\

\emph{\textbf{Example 3}}
We consider the BPSK with repetition code of order of $2$, and the $2 \times 1 \times 2$ backscatter RFID channel, i.e.
\begin{align}
\Delta=\left(
         \begin{array}{c}
           2  \\
           2  \\
         \end{array}
       \right),
\end{align}
and
\begin{align}
\mathbf{H}=\left(
         \begin{array}{c}
           h_{1,1}  \\
           h_{2,1}  \\
         \end{array}
       \right),
\mathbf{G}=\left(
         \begin{array}{cc}
           g_{1,1}   & g_{1,2} \\
         \end{array}
       \right).
\end{align}
In this case $M=2$, $L=1$, $N=2$,  $T=2$, $rank(\Delta)=1$, and $L_1^*=L_2^*=1$. We have
\begin{align}
R_{unitary}&=\sum_{t=1}^T \min(N,L_t^*)\nonumber\\
           &=\min(2,1)+\min(2,1)=2,
\end{align}
and
\begin{align}
R_{uniform}&=\min(N \times rank(\Delta), L)\nonumber\\
           &=\min(2 \times 1, 1)=1.
\end{align}
Thus we expect that the unitary query can bring significant gain. The simulation results of this example are shown in Fig. \ref{Fig: UnitaryVsUniformN2RepetionCodeBPSK}. We observe that, at mid SNR regimes, the unitary query brings about $10$ dB gain, and this gain is larger in higher SNR regimes. With the conventional uniform query, this level of improvement is only achievable when multiple antennas are equipped on the tag. In other words, for high performance RFID systems, the proposed unitary query can transfer the complexity requirement from the tag end to the reader end.

%We consider a system in which the tag is limited to have only one antenna, which possible for the purpose of low cost. In this case, what ever how many antennas are deployed at the reader query end and the reader receiving end, and what ever the coding scheme is used, the measure for the uniform query achieves only
%\begin{align}
%R_{uniform}&=\min(N\times rank (\Delta), L) \\
%           &=\min(N\times rank (\Delta), 1) \\
%           &=1,
%\end{align}
%in other words, the performance of the uniform query is limited by the single-antenna on the tag. Now we consider the unitary query, and the measure is
%\begin{align}
%R_{unitary}=\sum_{t=1}^T \min(N, L_t^*)\leq \sum_{t=1}^T \min(N, 1)=T,
%\end{align}
%clearly when $L_t^*=L=1$ for all $t$, the $R_{unitary}$ achieves maximum. It is clearly that a simple repetition code of order $T$ can achieve the about maximum measure. For instance, we consider the non-coherent 2FSK with repetition code of order of $2$, and the $2 \times 1 \times 2$ backscatter RFID channel. The measure for unitary query achieves maximum $R_{unitary}=2$. Thus we expect the unitary query can bring significant gain. This is shown in Fig. \ref{Fig: UnitaryVsUniformN2RepetionCode}. At BER of $10^{-2}$, unitary query brings about $5$ dB gain, and this gain increases to about $9$ dB at $10^{-3}$ and will be even larger at higher SNR regimes.

%\begin{align}
%R_{unitary}=\min(1,2)+\min(1,2)+\min(1,2)=3,
%\end{align}
%\begin{align}
%R_{uniform}=\min(2\times 1,2)=2,
%\end{align}

\begin{figure}
\centering
  \includegraphics[scale=0.66]{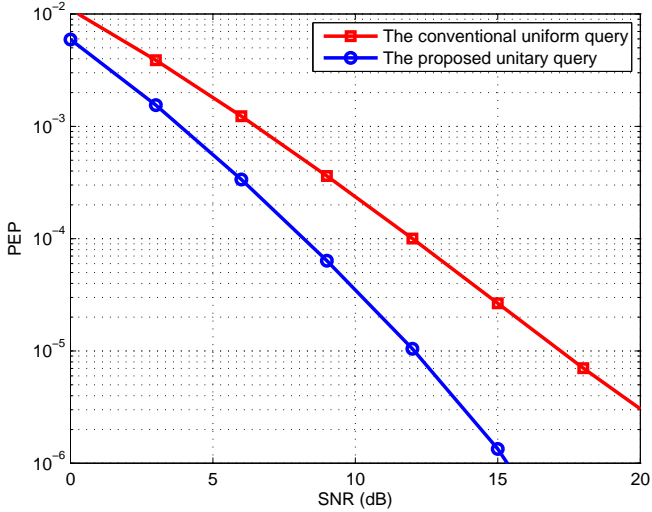}\\
  \caption{\emph{\textbf{Example 1}}: Performance comparisons between the unitary query and the uniform query. In this example, the space-time code that has a code difference matrix defined in \eqref{Eq: CodeWordDifferenceMatrix} is employed at the tag end. We can see that when the proposed unitary query is employed at the reader query end, the performance is much better than the case when the conventional unitary is employed. In the $2 \times 2 \times 2$ backscatter RFID channel, the unitary query can bring a significant gain: 5-7 dB improvement in the 10-20 dB SNR regime, this improvement can be even larger than 10 dB in higher SNR regimes. The simulations in this example agrees the performance measure in Theorem 1.}\label{Fig: UnitaryVsUniformN2}
\end{figure}

\begin{figure}
\centering
  \includegraphics[scale=0.66]{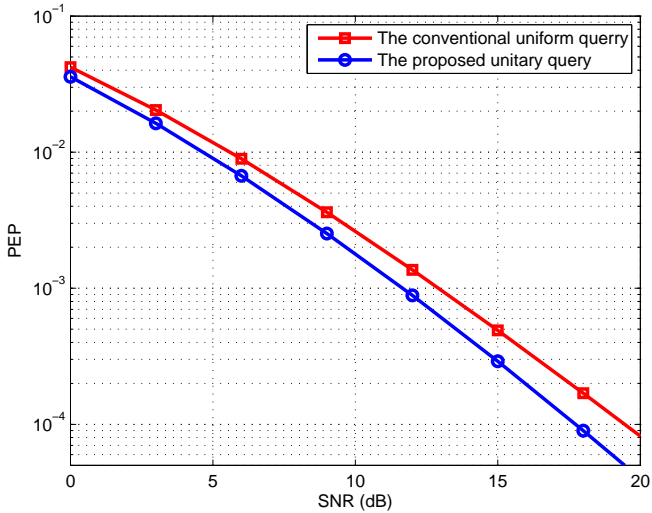}\\
  \caption{\emph{\textbf{Example 2}}: Performance comparisons between the unitary query and the uniform query. In this example, the space-time code which has a code difference matrix defined in \eqref{Eq: CodeWordDifferenceMatrix} is employed at the tag end. In the $2 \times 2 \times 1$ backscatter RFID channel, the unitary query outperforms the uniform query while the gain is not significant in this case. The simulations in this example agrees the performance measure in Theorem 1.}\label{Fig: UnitaryVsUniformN1}
\end{figure}

\begin{figure}
\centering
  \includegraphics[scale=0.66]{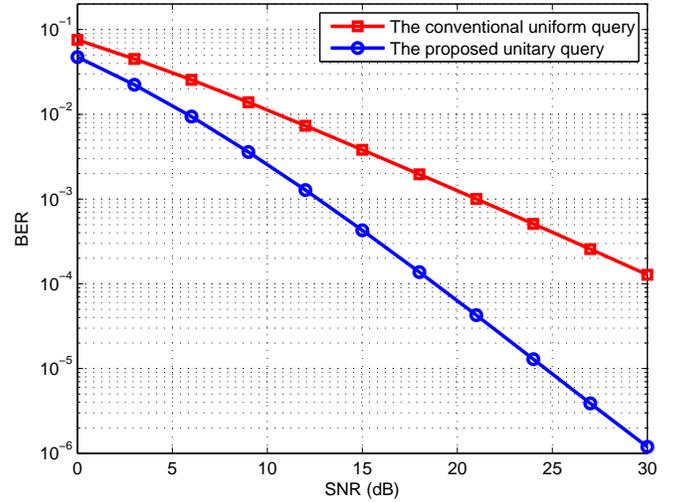}\\
  \caption{\emph{\textbf{Example 3}}: The unitary query can also significantly improve the performance of single-antenna tag. About $10$ dB gain brought by the unitary query is observed in the $2 \times 1 \times 2$ backscatter RFID channel in the mid SNR regimes when the tag employs repetition code. With the conventional uniform query, this level of improvement is only achievable when multiple antennas are equipped on the tag. }\label{Fig: UnitaryVsUniformN2RepetionCodeBPSK}
\end{figure}

%\begin{figure}
%\centering
%  % Requires \usepackage{graphicx}
%  \includegraphics[scale=0.66]{figuresUnitaryQuery/UnitaryVsUniformN1RepetionCode.eps}\\
%  \caption{BER performance comparisons between the unitary query and the uniform query for the $2 \times 1 \times 1$ backscatter RFID channel. Non-coherent binary FSK and non-coherent EGC are applied. The unitary query can bring a large gain.}\label{Fig: UnitaryVsUniformN1RepetionCode}
%\end{figure}

\section{Conclusion}\label{Sec: Conclusion}
In this paper, we proposed the unitary query at the reader query end in the $M \times L \times N$ MIMO backscatter RFID channel. We showed that even in the quasi-static fading, the unitary query can provide time diversity via multiple reader query antennas and thus can improve the performance of the RFID channel significantly. Due to the difficulty of evaluating the PEP and the diversity order directly, we derived a new performance measure based on the ranks of certain carefully constructed matrices.
Simulations showed that the proposed unitary query can improve the performance by $5$ to $10$ dB in mid-range SNR regimes, and the gain increases as the SNR increases. The unitary query can also improve the performance for the case of having single-antenna tag significantly, making it possible to employ inexpensive, small and low complex tags for high performance. In other words, for high performance RFID systems, the proposed unitary query can transfer the complexity requirement from the tag end to the reader end.
%To our best knowledge, this was the first time that the unitary query was proposed for the $M\times L \times N$ backscatter RFID channel.

% or


\section{Appendix}
\begin{proof}[\textbf{Proof of Lemma \ref{Lemma: Rank_of_Et}}]
Let $\mathbf{g}_1, \cdots, \mathbf{g}_N$ denote the columns of $\mathbf{G}$. We consider a set of scalars $\{a_1,\cdots,a_N\}$ where $a_n \in \mathbb{C}$,
for any linear combination of the set of vectors, $\{\mathbf{g}_1,\cdots,\mathbf{g}_N\}$,
\begin{align}
\mathbf{b}=\sum_{n=1}^L a_n\mathbf{g}_n
\end{align}
is a zero-mean complex Gaussian random vector with covariance matrix $\sum_{n=1}^L\|a_{n}\|^2 \mathbf{I}$,
therefore
\begin{align}\label{Eq: Linear_Combine_Columns_G}
\mathbb{P}\left(\mathbf{b}=\mathbf{0}\right)=0.
\end{align}
When $N\le L$, \eqref{Eq: Linear_Combine_Columns_G} implies that
\begin{align}
\mathbb{P}\left(\textmd{rank}(\mathbf{G})<N\right)=0,
\end{align}
or
\begin{align}
\mathbb{P}\left(\textmd{rank}(\mathbf{G})=N\right)=1.
\end{align}
 When $N > L$, by performing a linear combination of the rows of $\mathbf{G}$ and following a procedure similar to the case that $N\le L$, we can obtain
\begin{align}
\mathbb{P}\left(\textmd{rank}(\mathbf{G})=L\right)=1.
\end{align}
Hence the matrix $\mathbf{G}$ is of full rank with probability $1$, i.e.
\begin{align}
\mathbb{P}\left(\textmd{rank}(\mathbf{G})=\min(N,L)\right)=1.
\end{align}
Now notice that $\Delta_t$ is diagonal, therefore $\mathbf{E}_t=\Delta_t\mathbf{G}$ has $L_t^*$ non-zero rows.
Because $\mathbf{G}$ is full rank w.p. $1$,
we have
\begin{align}
rank(\mathbf{E}_t)=\min(L_t^*,N)
\end{align}
w.p. $1$.
\end{proof}
%%%%%%%%%%%%%%%%%%%%%%% Lemma 2 proof %%%%%%%%%%%%%%%%%%%%%%%%%%%%%%%
\begin{proof}[\textbf{Proof of Lemma \ref{Lemma: Rank_of_D}}]
Following similar steps to prove that $\mathbf{G}$ is of full rank w.p. $1$, we can show that
\begin{align}
\mathbb{P}(\textmd{rank}(\mathbf{G}_n)=L)=1,
\end{align}
i.e., $\mathbf{G}_n$ is also of full rank w.p. $1$.
Since
\begin{align}
\mathbf{D}_n=\Delta\mathbf{G}_{n},
\end{align}
we have
\begin{align}\label{Eq: Rank_D_n}
\mathbb{P}(\textmd{rank}(\Delta\mathbf{G}_n)=\textmd{rank}(\Delta))=1,
\end{align}
i.e. the rank of $\mathbf{D}_n$ is the same as the rank of $\Delta$ w.p. $1$.

Now let us consider the following two cases:\\
\emph{Case 1: $N \times rank(\Delta) \le L$}\\
By Eqn. \eqref{Eq: Rank_D_n}, clearly the columns of each of $\mathbf{D}_n$'s span a subspace of dimension $\textmd{rank}(\Delta)$ in $\mathbb{C}^{L}$ w.p. $1$.
Now consider a set of scalars $a_{i,j}$'s, where  $i \in \{1,\cdots,N\}$, and $j \in\{1,\cdots, T\}$.
If for $i \in \{2,\cdots,N\}$ and $j \in\{1,\cdots, T\}$, $a_{i,j}$'s are not all zero, it is not hard to verify that
\begin{align}
\mathbb{P}\left(\sum_{j=1}^T a_{1,j} \mathbf{D}_{1,j}=\sum_{i=2}^N\sum_{j=1}^T a_{i,j} \mathbf{D}_{i,j}\right)=0.
\end{align}
This implies that the rows of all $\mathbf{D}_n$'s span a subspace of dimension $N \times \textmd{rank}(\Delta)$ in $\mathbb{C}^{L}$ w.p. $1$, i.e. the rank of the block matrix $\mathbf{D}$ is $N \times \textmd{rank}(\Delta)$ w.p. 1 in this case.
\\
\emph{Case 2: $N \times rank(\Delta) > L$}\\
Following the similar procedure as in Case 1, we can find that the dimension of the subspace spanned by the columns of all $\mathbf{D}_n$'s is $L$, i.e. the rank of the block matrix $\mathbf{D}$ is $L$ w.p. 1 in this case.\\
With the results from Cases 1 and 2, we have Lemma \ref{Lemma: Rank_of_D} hold.
\end{proof}

%%%%%%%%%%%%%%%%%%%%%% poof of Theorem 1 begins %%%%%%%%%%%%%%%%%%%%%
\begin{proof}[\textbf{Proof of Theorem \ref{Th: New_Measure}}]
We consider singular value decompositions of $\mathbf{E}_t$'s and $\mathbf{D}$, i.e.,
\begin{align}
\mathbf{E}_t=\mathbf{U}_t\Lambda_t\mathbf{V}_t,
\end{align}
and
\begin{align}
\mathbf{D}=\mathbf{U}^*\Lambda^*\mathbf{V}^*.
\end{align}
Note that, for the unitary query, for a realization of $\mathbf{G}$ the squared distance between codewords can be given as
%%%%%%%%%%%%%%%%%%%%%%%%%%%%% Compute Z %%%%%%%%%%%%%%%%%%%%%%%%%%%%%%%%%%%%%%%%%%%%
\begin{align}
Z_X|\mathbf{G}&=\sum_{t=1}^T \|(x_{t,1},\cdots,x_{t,L})   \mathbf{E}_t \|_F^2 \nonumber\\
 &=\sum_{t=1}^T \|(x_{t,1},\cdots,x_{t,L})   \mathbf{U}_t\Lambda_t\mathbf{V}_t \|_F^2 \nonumber\\
 &\sim \sum_{t=1}^T \|(x_{t,1},\cdots,x_{t,L})   \Lambda_t \|_F^2 \nonumber\\
 &=\sum_{t=1}^T \sum_{i=1}^{\textmd{rank}(\mathbf{E}_t)} \lambda_{t,i}\|x_{t,i}\|^2,
\end{align}
where $\lambda_{t,i}$'s ($i=1,\cdots,\textmd{rank}(\mathbf{E}_t)$) are the non-zero eigenvalues of $\mathbf{E}_t$.
%%%%%%%%%%%%%%%%%%%%%%%%%%%%% Compute PEP_Z %%%%%%%%%%%%%%%%%%%%%%%%%%%%%%%%%%%%%%%%%%%%
Given a realization of $\mathbf{G}$, the conditional PEP on $\mathbf{G}$ is given by
\begin{align}\label{Eq: PEP_Z_G}
\text{PEP}_{Z_X|\mathbf{G}}(\bar{\gamma})
&=\mathbb{E}_{Z_X|\mathbf{G}}
\left(\mathbb{Q}\left(\sqrt{\bar{\gamma} \sum_{t=1}^T \sum_{i=1}^{rank(\mathbf{E}_t)} \lambda_{t,i}\|x_{t,i}\|^2/2}\right)\right) \nonumber\\
&=\prod_{t=1}^T\prod_{i=1}^{\textmd{rank}(\mathbf{E}_t)}\frac{1}{1+\lambda_{t,i}\bar{\gamma}/4}%\nonumber\\
%&=\frac{1}{\pi}\int_{\theta=0}^{\frac{\pi}{2}}
\end{align}
Therefore the PEP for the unitary query can be obtained as
\begin{align}\label{Eq: PEP_Z_G}
\text{PEP}_{Z_X}(\bar{\gamma})
&=\mathbb{E}_{\mathbf{G}}\left(\text{PEP}_{Z_X|\mathbf{G}}(\bar{\gamma})\right)\nonumber\\
&=\mathbb{E}_{\mathbf{G}}\left(\prod_{t=1}^T\prod_{i=1}^{\textmd{rank}(\mathbf{E}_t)}\frac{1}{1+\lambda_{t,i}\bar{\gamma}/4}\right)\nonumber\\
&=\mathbb{E}_{\mathbf{G}}\left(\prod_{t=1}^T\prod_{i=1}^{\min(N,L_t^*)}\frac{1}{1+\lambda_{t,i}\bar{\gamma}/4}\right)
%\nonumber\\
%&=\bar{\gamma}^{-R_{unitary}}\mathbb{E}_{\mathbf{G}}\left(\prod_{t=1}^T\prod_{i=1}^{\min(N,L_t^*)}\frac{1}{\frac{1}{\bar{\gamma}}+\lambda_{t,i}}\right)
\end{align}
The last step of the above derivation is obtained by using the result from Lemma \ref{Lemma: Rank_of_Et} and the fact that $0<\frac{1}{1+\lambda_{t,i}\bar{\gamma}/4}<\infty$.

%%%%%%%%%%%%%%%%%%%%%%%%%%%%% Compute Z^* %%%%%%%%%%%%%%%%%%%%%%%%%%%%%%%%%%%%%%%%%%%%
Similarly, for the uniform query, for a realization of $\mathbf{G}$, the squared distance between codewords can be given by
\begin{align}
Z_Y|\mathbf{G} & = \|(y_{1},\cdots,y_{L}) \mathbf{D}\|_F^2 \nonumber\\
    & = \|(y_{1},\cdots,y_{L}) \mathbf{U}^*\Lambda^*\mathbf{V}^*\|_F^2  \nonumber\\
    & \sim \|(y_{1},\cdots,y_{L}) \Lambda^*\|_F^2  \nonumber\\
    & = \sum_{i=1}^{rank(\mathbf{D})}\lambda_{i}^*\|(y_{1,i})\|^2,
\end{align}
where $\lambda_{i}^*$'s are the eigenvalues of $\mathbf{D}$.
%%%%%%%%%%%%%%%%%%%%%%%%% Compute PEP_Z* %%%%%%%%%%%%%%%%%%%%%%%%%%%%%%%%%%%%%%%%%%%%%%%
For a realization of $\mathbf{G}$, the conditional PEP is given by
\begin{align}\label{Eq: PEP_Zstar_G}
\text{PEP}_{Z_Y|\mathbf{G}}(\bar{\gamma})
&=\mathbb{E}_{Z_Y|\mathbf{G}}\left(\mathbb{Q}\left(\sqrt{\bar{\gamma}\sum_{i=1}^{\textmd{rank}(\mathbf{D})} \lambda_{i}^*\|x_{i}\|^2/2}\right)\right)\nonumber\\
&=\prod_{i=1}^{\textmd{rank}(\mathbf{D})}\frac{1}{1+\lambda_{i}^*\bar{\gamma}/4}
\end{align}
Therefore the PEP for the uniform query is given by
\begin{align}\label{Eq: PEP_Zstar_G}
\text{PEP}_{\mathbf{G}}(\bar{\gamma})
&=\mathbb{E}_{\mathbf{G}}\left(\prod_{i=1}^{\textmd{rank}(\mathbf{D})}\frac{1}{1+\lambda_{i}^*\bar{\gamma}/4}\right)\nonumber\\
&=\mathbb{E}_{\mathbf{G}} \left(\prod_{i=1}^{\min(N \times \textmd{rank}(\Delta),L)}\frac{1}{1+\lambda_{i}^*\bar{\gamma}/4}\right).
\end{align}
The last step of the above derivation is obtained by using  the result from Lemma \ref{Lemma: Rank_of_D} and the fact that $0<\frac{1}{1+\lambda_{i}^*\bar{\gamma}/4}<\infty$.
%%%%%%%%%%%%%%%%%%%%%%%%%%%%%%%%%%%%%%%%%%%%%%%%%%%%%%%%%%%%%%%%%%%%%%%%%%%%%%%%%
%%%%%%%%%%%%%%%%%% Compare PEP_Z and PEP_Z^* %%%%%%%%%%%%%%%%%%%%%%%%%%%%%%%%%%%%
%%%%%%%%%%%%%%%%%%%%%%%%%%%%%%%%%%%%%%%%%%%%%%%%%%%%%%%%%%%%%%%%%%%%%%%%%%%%%%%%%
With the assumption that
\begin{align}\label{Eq: FiniteExpectation1}
\mathbb{E}_{\mathbf{G}}\left(\prod_{t=1}^T\prod_{i=1}^{\min(N,L_t^*)}\frac{1}{\lambda_{i,t}}\right)<\infty,
\end{align}
\begin{align}\label{Eq: FiniteExpectation2}
\mathbb{E}_{\mathbf{G}}\left(\prod_{i=1}^{R_{\textmd{uniform}}}\frac{1}{\lambda_{i}^*}\right)<\infty,
\end{align}
%To have a measure for the above PEPs, instead, we consider the terms inside the two expectations first. It can be shown that if
%\begin{align}
%\sum_{t=1}^T \min(N,L_t^*)>\min(N\times\textmd{rank}(\Delta),L),
%\end{align}
%in asymptotic high SNR regimes, for any realization of $\mathbf{G}$ we have
%\begin{align}
%\lim_{\bar{\gamma} \rightarrow \infty}\prod_{t=1}^T\prod_{i=1}^{\min(N,L_t^*)}\frac{1}{1+\lambda_{t,i}\bar{\gamma}}
%<
%\lim_{\bar{\gamma} \rightarrow \infty}
%\prod_{i=1}^{\min(N\times\textmd{rank}(\Delta),L)}\frac{1}{1+\lambda_{i}^*\bar{\gamma}},
%\end{align}
%therefore we have
%\begin{align}
%&\lim_{\bar{\gamma}\rightarrow \infty}\mathbb{E}_{\mathbf{G}}\left(
%\prod_{t=1}^T\prod_{i=1}^{\min(N,L_t^*)}\frac{1}{1+\lambda_{i,t}\bar{\gamma}}\right)
%\\
%<&\lim_{\bar{\gamma}\rightarrow \infty}\mathbb{E}_{\mathbf{G}}\left(
%\prod_{i=1}^{\min(N\times\textmd{rank}(\Delta),L)}\frac{1}{1+\lambda_{i}^*\bar{\gamma}}\right),
%\end{align}
%or \eqref{Eq: Weaker_Measure} holds.
%Note that when
%\begin{align}
%\mathbb{E}_{\mathbf{G}}\left(\prod_{t=1}^T\prod_{i=1}^{\min(N,L_t^*)}\frac{1}{\lambda_{i,t}}\right)<\infty,
%\end{align}
and applying the Dominated Convergence Theorem (DCT), we have
\begin{align}\label{Eq: DCT_Unitary}
\lim_{\bar{\gamma}\rightarrow \infty}
\left(\bar{\gamma}^{R_{\textmd{unitary}}}\times\textmd{PEP}_{Z_X}(\bar{\gamma})\right)
=\mathbb{E}_{\mathbf{G}}\left(\prod_{t=1}^T\prod_{i=1}^{\min(N,L_t^*)}\frac{1}{\lambda_{i,t}}\right),
\end{align}
and
% similarly, using the assumption in \eqref{Eq: FiniteExpectation2} by applying DCT we have
%\begin{align}
%\mathbb{E}_{\mathbf{G}}\left(\prod_{i=1}^{R_{\textmd{uniform}}}\frac{1}{\lambda_{i}^*}\right)<\infty,
%\end{align}
\begin{align}\label{Eq: DCT_Uniform}
\lim_{\bar{\gamma}\rightarrow \infty}
\left(\bar{\gamma}^{R_{\textmd{uniform}}}\times\textmd{PEP}_{Z_Y}(\bar{\gamma})\right)
=\mathbb{E}_{\mathbf{G}}\left(\prod_{i=1}^{R_{\textmd{uniform}}}\frac{1}{\lambda_{i}^*}\right).
\end{align}
%therefore if $\mathbb{E}_{\mathbf{G}}\left(\prod_{t=1}^T\prod_{i=1}^{\min(N,L_t^*)}\frac{1}{\lambda_{i,t}}\right)$ and $\mathbb{E}_{\mathbf{G}}\left(\prod_{i=1}^{R_{\textmd{uniform}}}\frac{1}{\lambda_{i}^*}\right)$ are finite we have
%%%%%%% Case 1 %%%%%%%%%%%%%%%%%%%%%%%%%%%%%%%%%%%%%%%%%%%%%%%%%%%%%%%%%%%%%%%%%
\emph{Case 1}:$R_{\textmd{unitary}}>R_{\textmd{uniform}}$\\
In this case,
\begin{align}
\lim_{\bar{\gamma}\rightarrow \infty}
\frac{\textmd{PEP}_{Z_X}(\bar{\gamma})}{\text{PEP}_{Z_Y}(\bar{\gamma})}
&=\lim_{\bar{\gamma}\rightarrow \infty}\frac{\bar{\gamma}^{R_{\textmd{uniform}}}\mathbb{E}_{\mathbf{G}}\left(\prod_{t=1}^T\prod_{i=1}^{\min(N,L_t^*)}\frac{1}{\lambda_{i,t}}\right)}
{\bar{\gamma}^{R_{\textmd{unitary}}}\mathbb{E}_{\mathbf{G}}\left(\prod_{i=1}^{R_{\textmd{uniform}}}\frac{1}{\lambda_{i}^*}\right)}\nonumber\\
&\rightarrow 0.%\lim_{\bar{\gamma}\rightarrow \infty}\frac{\bar{\gamma}^{R_{\textmd{uniform}}}}
%{\bar{\gamma}^{R_{\textmd{unitary}}}}=0
\end{align}
%%%%%%% Case 2 %%%%%%%%%%%%%%%%%%%%%%%%%%%%%%%%%%%%%%%%%%%%%%%%%%%%%%%%%%%%%%%%%
\emph{Case 2}:$R_{\textmd{unitary}}<R_{\textmd{uniform}}$\\
In this case,
\begin{align}
\lim_{\bar{\gamma}\rightarrow \infty}
\frac{\textmd{PEP}_{Z_Y}(\bar{\gamma})}{\text{PEP}_{Z_X}(\bar{\gamma})}
&=\lim_{\bar{\gamma}\rightarrow \infty}\frac{\bar{\gamma}^{R_{\textmd{unitary}}}\mathbb{E}_{\mathbf{G}}\left(\prod_{i=1}^{R_{\textmd{uniform}}}\frac{1}{\lambda_{i}^*}\right)}
{\bar{\gamma}^{R_{\textmd{uniform}}}\mathbb{E}_{\mathbf{G}}\left(\prod_{t=1}^T\prod_{i=1}^{\min(N,L_t^*)}\frac{1}{\lambda_{i,t}}\right)}
\nonumber\\
&\rightarrow 0.%\lim_{\bar{\gamma}\rightarrow \infty}\frac{\bar{\gamma}^{R_{\textmd{uniform}}}}
%{\bar{\gamma}^{R_{\textmd{unitary}}}}=0
\end{align}
%%%%%%% Case 3 %%%%%%%%%%%%%%%%%%%%%%%%%%%%%%%%%%%%%%%%%%%%%%%%%%%%%%%%%%%%%%%%%
\emph{Case 3}:
$R_{\textmd{unitary}}=R_{\textmd{uniform}}$\\
In this case, we have
\begin{align}
\lim_{\bar{\gamma}\rightarrow \infty}
\frac{\textmd{PEP}_{Z_X}(\bar{\gamma})}{\text{PEP}_{Z_Y}(\bar{\gamma})}
&=\lim_{\bar{\gamma}\rightarrow \infty}\frac{\bar{\gamma}^{R_{\textmd{uniform}}}\mathbb{E}_{\mathbf{G}}\left(\prod_{t=1}^T\prod_{i=1}^{\min(N,L_t^*)}\frac{1}{\lambda_{i,t}}\right)}
{\bar{\gamma}^{R_{\textmd{unitary}}}\mathbb{E}_{\mathbf{G}}\left(\prod_{i=1}^{R_{\textmd{uniform}}}\frac{1}{\lambda_{i}^*}\right)}\nonumber\\
&=\frac{\mathbb{E}_{\mathbf{G}}\left(\prod_{t=1}^T\prod_{i=1}^{\min(N,L_t^*)}\frac{1}{\lambda_{i,t}}\right)}
{\mathbb{E}_{\mathbf{G}}\left(\prod_{i=1}^{R_{\textmd{uniform}}}\frac{1}{\lambda_{i}^*}\right)}
=c.
\end{align}
\end{proof}

\ifCLASSOPTIONcaptionsoff
  \newpage
\fi



%\bibliographystyle{IEEEtran}

\bibliographystyle{IEEEtran}
\bibliography{RFIDarticle2}

\end{document}